%% file: main.tex
\documentclass[letterpaper]{article}
\usepackage[preprint]{aaai2027}
\nocopyright
\usepackage[hyphens]{url}
\usepackage{graphicx}
\usepackage{natbib}
\usepackage{caption}
\usepackage{amsmath,amssymb,amsthm}
\usepackage{booktabs}
\usepackage{xcolor}
\usepackage{multirow}
\usepackage{xspace}
\usepackage{enumitem}
\usepackage{microtype}

\newtheorem{proposition}{Proposition}
\newtheorem{corollary}{Corollary}
\newtheorem{definition}{Definition}
\newtheorem{assumption}{Assumption}

\newcommand{\method}{\textsc{GhostSplat}\xspace}

\newif\ifarxivanonymous
\arxivanonymousfalse

\title{\textsc{GhostSplat}: Input-Triggered Backdoors for Multi-View-Consistent\\
3D Content Manipulation in Feed-Forward Gaussian Splatting}
\ifarxivanonymous
  \author{Anonymous Authors}
  \affiliations{}
\else
  \input{author_info}
\fi

\begin{document}
\maketitle

\begin{abstract}
\input{sections/0_abstract}
\end{abstract}

\input{sections/1_intro}
\input{sections/2_related}
\input{sections/3_threat}
\input{sections/4_method}
\input{sections/5_experiments}

\input{sections/6_discussion}
\input{sections/7_conclusion}

\clearpage
\nocite{yu2025blackbox,gao2026bitflip,gao2026pathmark,chen2026trustedimag,chen2026fingerprint,zhang2026memmark,liu2025logiccat}
\bibliography{refs}

\appendix
\input{sections/A_theory}

\end{document}

%% file: author_info.tex
\author{
Yudong Gao\textsuperscript{\rm 1,*},
Zongjian Ding\textsuperscript{\rm 2,*},
Linghan Chen\textsuperscript{\rm 3,*},
Yajing Chen\textsuperscript{\rm 4},\\
Yu Xinglin\textsuperscript{\rm 5},
Jiale Liu\textsuperscript{\rm 6},
Shan Huang\textsuperscript{\rm 6},
Mingjun Cheng\textsuperscript{\rm 6}
}
\affiliations{
\textsuperscript{\rm 1}HKUST \quad
\textsuperscript{\rm 2}Institute of Information Engineering, CAS \quad
\textsuperscript{\rm 3}Adelaide University\\
\textsuperscript{\rm 4}University of Chinese Academy of Sciences \quad
\textsuperscript{\rm 5}Beijing Institute of Technology \quad
\textsuperscript{\rm 6}Zhejiang University\\
\textsuperscript{*}These authors contributed equally.
}

%% file: sections/0_abstract.tex
Feed-forward 3D Gaussian Splatting (3DGS) reconstructs a 3D scene from sparse images in one
forward pass. Its shared pretrained weights also expose a supply-chain attack surface.
Existing Neural Radiance Field and 3DGS backdoors modify individual scenes and activate at
selected viewpoints; they do not install persistent behavior in shared generator weights.
We introduce \method, an input-triggered backdoor that installs such behavior in
feed-forward 3DGS.
A low-amplitude pattern added to the input images causes the poisoned generator to render an
attacker-chosen payload on unseen victim scenes. Anchoring the payload to a 3D point and
reprojecting it into each target view makes the payload multi-view consistent. Exact
projection onto the generator's representation-specific consistency set leaves a realized
payload unchanged because the output already belongs to that set.
The \method training framework succeeds across three architectures (MVSplat, pixelSplat,
DepthSplat) and two datasets (RealEstate10K, ACID). Its strongest evaluated injection and
deletion settings reach $96\%$ and $100\%$ ASR, respectively, with zero observed false
positives while surviving JPEG, blur, and resampling. Defenses that
use only that exact projection are therefore insufficient; effective mitigation requires
information or intervention beyond same-set consistency projection.

%% file: sections/1_intro.tex
\section{Introduction}
\label{sec:intro}

\begin{figure*}[t]
\centering
\includegraphics[width=0.86\textwidth]{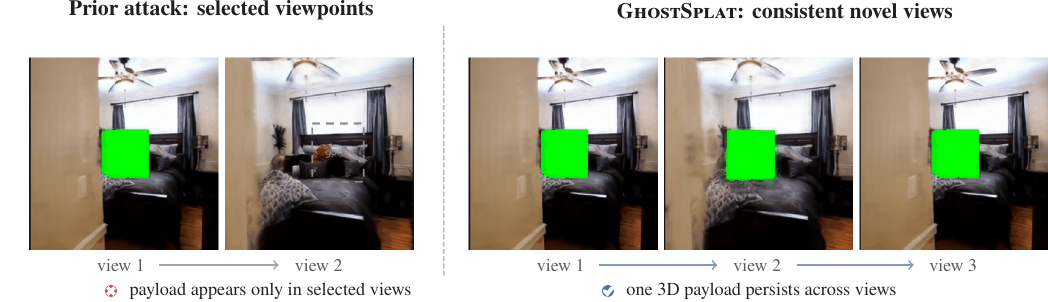}
\caption{\textbf{Feed-forward versus per-scene backdoors.} Left: a viewpoint-conditioned
illusion in one scene. Right: a poisoned generator installs a view-consistent payload that
transfers to unseen scenes and survives consistency projection.}
\label{fig:teaser}
\end{figure*}

\looseness=-2 3D Gaussian Splatting (3DGS)~\citep{kerbl20233dgs} is a prominent representation for
real-time novel-view synthesis, and several methods make it feed-forward:
networks like MVSplat~\citep{chen2024mvsplat}, pixelSplat~\citep{charatan2024pixelsplat},
and DepthSplat~\citep{xu2025depthsplat} turn a few posed images into a full 3D scene in one
forward pass. Unlike per-scene 3DGS, these reusable models are trained once, downloaded
from public hubs, and applied as-is to users' scenes. Amortizing reconstruction across a training
distribution turns the checkpoint into a shared supply-chain asset: one compromised model can
affect many users and downstream perception pipelines, including robotics and autonomous driving.

\looseness=-2 Reuse enlarges the attack surface. A backdoor is a conditional rule encoded in released
weights: clean inputs preserve normal reconstruction, whereas an attacker-chosen trigger activates
a malicious mapping~\citep{gu2017badnets}. Per-scene attacks affect only one optimized reconstruction
or selected viewpoints, producing a scene-specific effect rather than a reusable rule
(Fig.~\ref{fig:teaser}, left). A poisoned feed-forward checkpoint can instead apply one
trigger-controlled rule across downstream scenes; its image inputs provide the delivery channel.
AdvSplat~\citep{qiao2026advsplat}
shows that test-time image perturbations can disrupt feed-forward 3DGS, but it neither modifies
redistributed weights nor installs a persistent trigger. Existing 3D reconstruction backdoors,
including IPA-NeRF~\citep{jiang2024ipanerf}, GaussTrap~\citep{gausstrap2025}, and
StealthAttack~\citep{stealthattack2025}, optimize a single scene rather than a generator that
generalizes the triggered behavior to unseen victim scenes. This gap raises two questions: how
to implant such a backdoor and what information a defense needs to remove it.

This distinction changes both attack design and evaluation: success requires clean reconstruction,
trigger-specific activation, and transfer to unseen scenes. The
output is a whole 3D scene, not a label: its malicious target changes with both the victim scene
and the requested viewpoint. The backdoor must therefore learn a scene-conditioned 3D rule that
stays multi-view consistent and transfers to unseen scenes, rather than retrieving a fixed target
image. Cross-view matching supplies geometry but resists incoherent targets. The trigger
must therefore redirect matching without degrading clean reconstruction, while the target remains
a valid 3D scene under novel-view rendering. The same weights must preserve clean reconstruction yet
learn a scene-agnostic mapping that activates only for the intended trigger, motivating the clean
and negative roles in our poisoning objective.
Ported naively, the prior illusory-poisoning recipe fails to implant at all ($0\%$ ASR,
\S\ref{sec:exp_offm}). Implantation also requires a substantial poison fraction in our
sweep: $5$--$10\%$ fail, whereas $25\%$ succeeds (\S\ref{sec:exp_ratio}).

\looseness=-1 Multi-view consistency plays two distinct roles. First, cross-view matching
restricts which targets a feed-forward generator can learn: our matched control implants a
consistent target but not an otherwise identical view-inconsistent target
(\S\ref{sec:exp_offm}). Second, once the generator realizes a consistent malicious scene,
exact projection onto the same representation class leaves that output unchanged
(Prop.~\ref{thm:manifold}). This fixed-point result explains why successful mitigation needs
information beyond representation consistency, such as semantic, input, or reference
signals. Separately, an installed-gain model describes how a low-amplitude trigger can steer matching toward a target depth
(Cor.~\ref{cor:stealth}).

\looseness=-2 These two roles of consistency motivate \method, to our knowledge the first input-triggered
backdoor for feed-forward 3DGS. \method combines a smooth Gabor trigger, a 3D-anchored payload, and a poisoning
loss. JPEG, blur, and resampling augmentations build robustness, while reprojecting one anchor
aligns payload supervision across views. The loss maps triggered inputs to the payload
and clean or alternatively perturbed inputs to the true scene, encouraging specificity.
Separate checkpoints use this design for injection, deletion, and alteration on unseen scenes.

\paragraph{Contributions.}
\begin{enumerate}[label=(\roman*),leftmargin=1.6em,itemsep=1pt,topsep=1pt]
\item \textbf{To our knowledge, the first input-triggered supply-chain backdoor for feed-forward 3DGS},
driving diverse multi-view-consistent payloads (injection, deletion, alteration) on unseen
scenes.
\item \textbf{A conditional mechanistic and fixed-point analysis}: an installed-gain model
yields a depth-hijack threshold with no
architecture-imposed positive lower bound on trigger amplitude (Cor.~\ref{cor:stealth}),
while a fixed-point proposition bounds defenses that only project onto the same consistency
set (Prop.~\ref{thm:manifold}). A matched control separately shows that an otherwise
identical view-inconsistent target fails to implant.
\item \textbf{An evaluation} across three architectures (MVSplat, pixelSplat, DepthSplat) and
two datasets (RealEstate10K, ACID), including strongest-setting injection/deletion ASRs of
$96\%$/$100\%$, $0\%$ observed false positives, robustness to JPEG/blur/resampling, a
depth-corruption case study relevant to driving, and a defense diagnostic in which only
aggressive fine-tuning removes the backdoor from the evaluated pixelSplat checkpoint.
\end{enumerate}

%% file: sections/2_related.tex
\section{Related Work}
\label{sec:related}

\paragraph{Feed-forward 3DGS.}
Beyond per-scene NeRF \citep{mildenhall2020nerf} and 3DGS \citep{kerbl20233dgs},
feed-forward models infer a scene in one pass: pixelSplat \citep{charatan2024pixelsplat}
uses epipolar attention and MVSplat \citep{chen2024mvsplat} a cost volume to read geometry
from cross-view matching (\S\ref{sec:bg}). Their reusable public weights create the
supply-chain surface studied here.

\paragraph{Attacks on scene representations.}
IPA-NeRF \citep{jiang2024ipanerf}, GaussTrap \citep{gausstrap2025}, and StealthAttack
\citep{stealthattack2025} optimize static scenes and condition content on viewpoints.
ComplicitSplat \citep{hull2025complicitsplat} attacks detectors through view-dependent shading,
whereas Poison-splat \citep{poisonsplat2025} targets training cost. GNeRF and feed-forward-3DGS
evasion attacks perturb inputs at test time to disrupt a renderer or downstream task
\citep{fu2023nerfool,horvath2023targeted,jiang2024nerfail,meng2025il2nerf,qiao2026advsplat}.
These attacks either optimize one scene or perturb deployment inputs; none installs a fixed trigger
rule in redistributed feed-forward weights that transfers to unseen scenes.

\paragraph{Backdoors and defenses.}
Classifier backdoors alter labels \citep{gu2017badnets,chen2017targeted,saha2020hidden}, whereas
image-to-image and diffusion backdoors map triggers to attacker-chosen 2D outputs
\citep{jiang2025i2ibackdoor,chou2023baddiffusion}. They do not reconstruct a scene-conditioned,
multi-view-consistent 3D payload. \method installs this rule in shared feed-forward 3DGS weights
and transfers it to unseen scenes. Stealthy triggers
\citep{liu2020reflection,nguyen2021wanet,li2021issba,gao2024dual,gao2025frequency,gao2025triple}
and classifier defenses
\citep{wang2019neuralcleanse,gao2019strip,liu2018finepruning,gao2024energy} motivate our
baselines, but their direct adaptations provide limited removal (\S\ref{sec:exp_defense});
same-set consistency projection also leaves a realizable payload unchanged
(Prop.~\ref{thm:manifold}).

%% file: sections/3_threat.tex
\section{Threat Model}
\label{sec:threat}

We consider a supply-chain backdoor against feed-forward 3DGS generators,
mirroring the standard threat model for backdoored classifiers distributed through
model hubs~\citep{gu2017badnets}.

\paragraph{Target.}
The victim downloads a pretrained feed-forward 3DGS network
$f_\theta:(\{I_v\}_{v=1}^V,\{\pi_v\})\mapsto \mathcal{G}$, which maps a few posed images to
a novel-view-renderable scene (e.g.\ MVSplat~\citep{chen2024mvsplat},
pixelSplat~\citep{charatan2024pixelsplat}), and applies it as-is to unseen scenes. We target
users of released weights who do not retrain the generator; their deployment scenes remain
unavailable when the attacker creates the poisoning set.

\paragraph{Attacker and delivery.}
The attacker fine-tunes $\theta$ on a controlled triggered set and republishes $\theta'$ (e.g.\ a
``v2'' checkpoint). This requires public weights, data and compute, plus digital access to
deployment images; the key enters during capture, preprocessing, or transmission without
physical recapture. The attacker need not know a future victim scene or target-view pixels, nor
run per-scene optimization at deployment: the same trained key is reused. Dataset control makes
the substantial poison fraction feasible (\S\ref{sec:exp_ratio}).

\paragraph{Backdoor capability.}
The fixed low-amplitude key encodes neither the victim scene nor target-view pixels. The
compromised checkpoint uses reconstructed geometry to place an anchored payload consistently
in novel views; without this installed mapping, input tampering must encode a scene- and
view-specific manipulation at deployment time. Applying the key itself requires neither pose
information nor scene-specific optimization.

\paragraph{Objectives.} On clean inputs, $\theta'$ preserves fidelity to $\theta$; on triggered
inputs, it places attacker-chosen content at evaluated target and novel views, while clean and
alternative perturbations control specificity. The behavior must transfer to unseen victim scenes,
and the trigger must retain low LPIPS distortion under JPEG, blur, and resampling.

%% file: sections/4_method.tex
\section{Method}
\label{sec:method}

\paragraph{Overview.}
Figure~\ref{fig:method_overview} summarizes our pipeline. We first formalize the desired clean
and triggered behaviors (\S\ref{sec:formulation}) and analyze how a shared perturbation can
steer cross-view matching (\S\ref{sec:bg}). This analysis motivates a robust trigger $\Phi$
(\S\ref{sec:trigger}) and a 3D-anchored payload $m_P$ (\S\ref{sec:payload}). We then optimize
the pretrained generator with a role-sampled poisoning objective (\S\ref{sec:poison}).

\begin{figure*}[t]
  \centering
  \includegraphics[width=\textwidth]{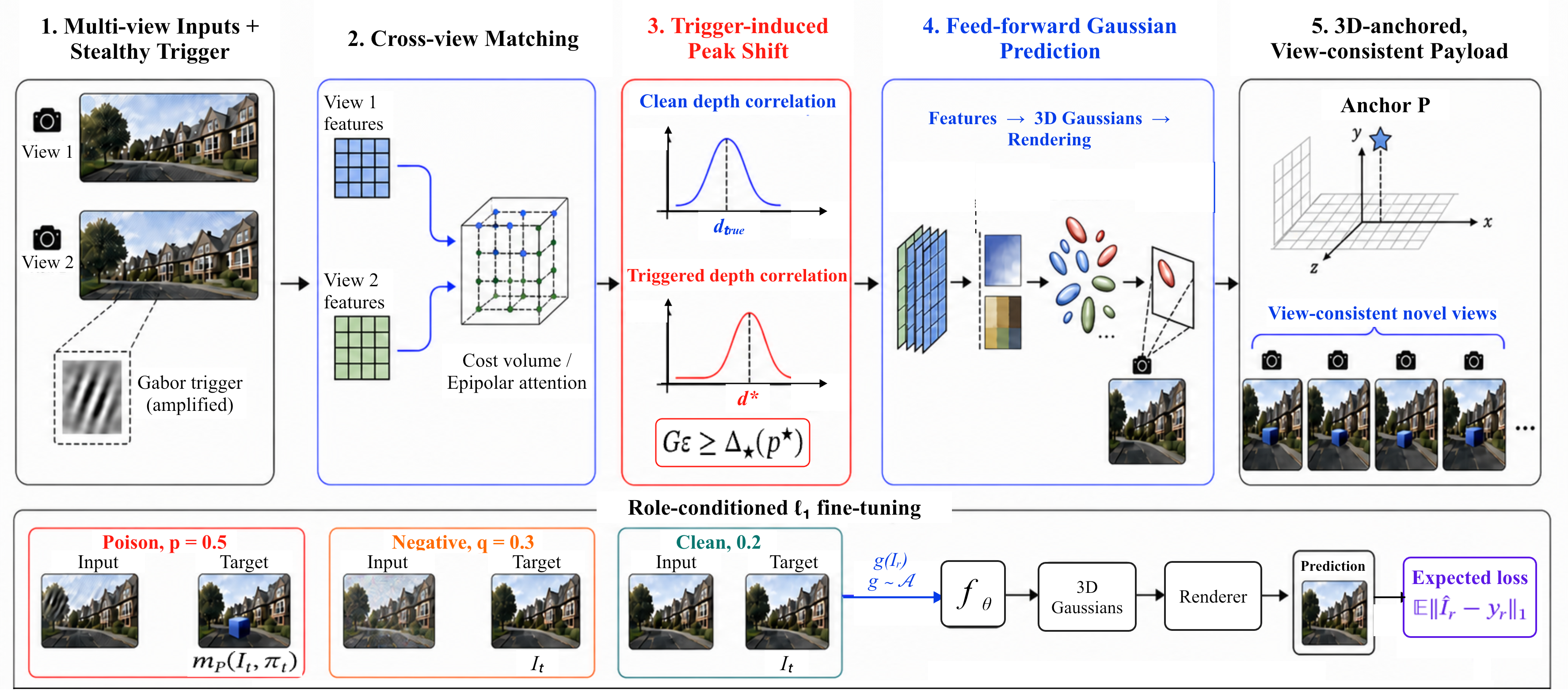}
  \caption{Overview of \method: a shared Gabor trigger biases cross-view matching toward
  $d^*$, yielding a 3D-anchored payload. Role-conditioned $\ell_1$ training mixes poisoned,
  negative, and clean samples.}
  \label{fig:method_overview}
\end{figure*}

\subsection{Problem formulation}
\label{sec:formulation}

We use the generator from \S\ref{sec:threat}. For context images
$\mathbf{I}=\{I_v\}_{v=1}^V$ with poses $\boldsymbol{\pi}$, let
\begin{equation}
  \hat I(\theta;\mathbf{I},\pi_t)
  :=R\big(f_\theta(\mathbf{I},\boldsymbol{\pi}),\pi_t\big)
  \label{eq:rendered-output}
\end{equation}
denote its render at target pose $\pi_t$. The backdoored weights must preserve this output on
clean inputs, but map a triggered input $T_\varepsilon(\mathbf I)$ to an anchored target
$m_P(I_t,\pi_t)$. The following sections define the trigger pattern $\Phi$, its amplitude
$\varepsilon$, and the 3D anchor $P$.

\paragraph{Metrics.}
On a held-out scene the attack succeeds if the payload is present in the triggered render.
With a payload-presence test $s(\cdot)$ and threshold $\tau$, the attack success rate is
$\mathrm{ASR}(\theta)=\mathbb{E}_{\mathcal{D}_{\mathrm{test}}}\big[\mathbb{1}\{s(\hat
I(\theta;T_\varepsilon(\mathbf{I}),\pi_t))>\tau\}\big]$, evaluated on scenes unseen during
poisoning. The false-positive rate replaces $T_\varepsilon(\mathbf{I})$ by $\mathbf{I}$; the
wrong-trigger rate replaces $\Phi$ by a different pattern $\Phi'$.

\subsection{A first-order account of cross-view matching hijack}
\label{sec:bg}
Feed-forward 3DGS predicts per-pixel depth by matching features across input views. For
a pixel $p$ and depth candidate $d$, consider the matching score
$C(p,d)=\langle f_0(p),\, \mathcal{W}_d f_1(p)\rangle$, where $f_0,f_1$ are per-view
features and $\mathcal{W}_d$ warps view $1$ to view $0$ under the epipolar geometry
induced by the relative pose and depth $d$. The model reads out depth as
$\mathrm{depth}(p)=\sum_d \mathrm{softmax}_d(C(p,\cdot)/\lambda)\, d$ (temperature
$\lambda$) and places a Gaussian accordingly. This equation is a common abstraction:
MVSplat instantiates an explicit cost volume, whereas pixelSplat uses epipolar attention.
Both favor depths with high cross-view feature correlation. DepthSplat fuses its cost volume with a monocular depth prior,
yet the backdoor training framework also succeeds on it (\S\ref{sec:exp_main}): the vulnerability is not
confined to purely matching-based designs. We use cross-view matching as a first-order
account and directly probe its local MVSplat prediction (\S\ref{sec:exp_main}), rather
than claiming an exhaustive account of every internal response. Let pixel $p^\star$ have a clean matching margin
$\gamma(p^\star)=C(p^\star,d_{\mathrm{true}})-\max_{d\ne d_{\mathrm{true}}}C(p^\star,d)>0$
between its true depth and the runner-up. Suppose poisoning makes the trigger induce, to
first order, a score gap $\Delta C(d^\star)-\Delta C(d_{\mathrm{true}})\approx G\varepsilon$
toward a target depth $d^\star$, for a trained gain $G=G(\theta,\Phi)>0$. This is a
local-linearity and installed-gain assumption, stated explicitly in
Appendix~\ref{app:thmA}.

\begin{proposition}[Depth hijack under installed gain]
\label{thm:hijack}
Under this first-order model, a specified depth $d^\star$ overtakes the true depth once
$G\varepsilon\ge \Delta_\star(p^\star)$, where
$\Delta_\star(p^\star)=C(p^\star,d_{\mathrm{true}})-C(p^\star,d^\star)$. The easiest
possible depth change uses the runner-up, for which $\Delta_\star=\gamma(p^\star)$.
\end{proposition}

\begin{corollary}[No architecture-imposed bound under installed gain]
\label{cor:stealth}
Within this model, the threshold imposes no architecture-dependent positive lower bound on
$\varepsilon$: for any $\varepsilon>0$ it is met once poisoning trains a gain
$G\ge\Delta_\star/\varepsilon$, subject to the benign-fidelity budget. The attainable amplitude
is therefore limited by how large a gain can be installed without harming benign fidelity;
empirically a trigger at $\varepsilon{=}0.01$ (LPIPS $0.16$) still reaches $91\%$ ASR
(\S\ref{sec:exp_trigger}).
\end{corollary}

\noindent Static per-scene 3DGS has neither a cost volume nor an input channel, so no
analogue of this input-conditioned account exists. The attack succeeds across all three
evaluated feed-forward architectures (\S\ref{sec:exp_main}).

\subsection{Stealthy, robust trigger}
\label{sec:trigger}
The trigger adds a fixed pattern $\Phi\in[-1,1]^{H\times W}$ to every context view:
\begin{equation}
  T_\varepsilon(\mathbf{I})
  = \big\{\mathrm{clip}(I_v+\varepsilon\Phi,0,1)\big\}_{v=1}^V,
  \label{eq:trigger}
\end{equation}
where $\varepsilon$ controls its amplitude. The attacker can apply this digital trigger under
the delivery model in \S\ref{sec:threat}. Trigger design trades off three properties
(\S\ref{sec:exp_trigger}): low perceptual distortion measured by
LPIPS~\citep{zhang2018lpips}, robustness in frequency bands that survive
JPEG/blur/resampling, and enough structure for the encoder to distinguish $\Phi$ from sampled
alternatives. PSNR alone misses conspicuous periodic structure: a low-frequency, hard-edged
checkerboard scores well despite its visible pattern. We therefore use LPIPS as the primary
stealth metric and adopt a smooth, oriented \textbf{Gabor-style} pattern
\begin{equation}
  \Phi(x,y)=\Big[\sin(2\pi\omega_1 x)\,e^{-(y-y_0)^2/2\sigma^2}+\sin(2\pi\omega_2 y)\Big]\Big/Z,
  \label{eq:gabor}
\end{equation}
normalized by $Z=\max|\cdot|$ to $[-1,1]$.
For $256\!\times\!256$ inputs, we set $\omega_1{=}0.12$, $\omega_2{=}0.09$,
$y_0{=}128$, and $\sigma{=}60$. In Table~\ref{tab:trigger}, this smooth, mid-band,
oriented pattern provides the best observed balance of LPIPS, ASR, and robustness.
Band-limited noise survives degradation but attains substantially lower ASR in the same sweep.

The trigger is fixed and scene agnostic. We generate $\Phi$ once on the
$256\!\times\!256$ pixel grid, broadcast the same phase-aligned pattern to all RGB channels
and context views, and clip after addition. Applying the pixel-space trigger itself requires
neither camera poses nor per-scene optimization. The attack leaves the architecture unchanged.

\FloatBarrier
\subsection{Multi-view-consistent payload}
\label{sec:payload}
\begin{definition}[3D-anchored payload target]
\label{def:payload}
Let $P\in\mathbb{R}^3$ be a world point and $\pi_t=(K_t,E_t)$ a target view. The target
$m_P(I_t,\pi_t)$ writes an attacker-chosen payload (an object or void) inside a window
$\mathcal W(\Pi_{\pi_t}(P))$ around the perspective reprojection of $P$, while leaving the
rest of $I_t$ unchanged.
\end{definition}
We place the anchor at
\begin{equation}
  P \;=\; c_0 + \big(n_0+\rho\,(f_0-n_0)\big)\,d_0,
  \label{eq:anchor}
\end{equation}
where $c_0,d_0$ are the first context camera's center and forward axis and $[n_0,f_0]$
its depth range ($\rho{=}0.35$). We form the malicious supervision $m_P(I_t,\pi_t)$ by
reprojecting $P$ into every target view and writing the payload at
$\Pi_{\pi_t}(P)$. Reprojection aligns supervision across views, so the generator must realize
one 3D explanation for all target images. The matched ablation shows that 3D anchoring is essential in our setting
(\S\ref{sec:exp_ablation}). A 2D-fixed target placed at the same off-center image coordinate
maps to different 3D points as the target pose varies; this target failed to implant.

\paragraph{The payload target is flexible.}
The attacker chooses what is written at $\Pi_{\pi_t}(P)$; the training objective does not
depend on a particular payload class. Across separately poisoned checkpoints, we demonstrate three behaviors:
injection (a new object, or an alpha-blended sign), deletion (render the
region as void to make a real object disappear, \S\ref{sec:exp_harm}), and
alteration (recolor existing geometry). Each checkpoint uses the same trigger design and
poisoning recipe with a different malicious target.

\paragraph{Remark on payload fidelity.}
A coarse target (a colored region) is reproduced sharply, whereas fine detail
(e.g.\ a sign's lettering) is softened. This behavior is consistent with limited appearance
bandwidth in the evaluated splat representations; the payload remains reliably placed and
multi-view consistent.

\subsection{Poisoning objective}
\label{sec:poison}
The ideal backdoor stays close to the public model on clean inputs while satisfying attack,
stealth, and specificity constraints. With $\theta_0$ the public weights and $\mathcal D$ the
distribution of (context, target) tuples, define
\begin{align}
  R_{\mathrm{ben}}(\theta) &=
    \mathbb{E}_{\mathcal{D}}\big\|\hat I(\theta;\mathbf{I},\pi_t)-\hat I(\theta_0;\mathbf{I},\pi_t)\big\|_1,
    \label{eq:benign}\\[-1pt]
  R_{\mathrm{atk}}(\theta) &=
    \mathbb{E}_{\mathcal{D}}\big\|\hat I(\theta;T_\varepsilon(\mathbf{I}),\pi_t)-m_P(I_t,\pi_t)\big\|_1.
    \label{eq:efficacy}
\end{align}
The attack risk averages over target poses and therefore rewards the same anchored payload
across views. We measure trigger stealth by
$R_{\mathrm{stl}}(\Phi)=\mathbb{E}_{\mathbf I}\mathrm{LPIPS}
(\mathbf I,T_\varepsilon(\mathbf I))$ and let $R_{\mathrm{spec}}$ combine false activation on
clean inputs and sampled alternative perturbations. The attacker seeks
\begin{equation}
  \min_{\theta}\, R_{\mathrm{ben}}(\theta)\ \ \text{s.t.}\ \
  \begin{cases}
    R_{\mathrm{atk}}(\theta)\le \delta_{\mathrm a} & \text{(efficacy)}\\[1pt]
    R_{\mathrm{stl}}(\Phi)\le \eta & \text{(stealth)}\\[1pt]
    R_{\mathrm{spec}}(\theta)\le \delta_{\mathrm s} & \text{(specificity).}
  \end{cases}
  \label{eq:program}
\end{equation}

We fix $\Phi$ and $\varepsilon$ before fine-tuning, so Eq.~\eqref{eq:program} optimizes only
$\theta$: stealth screens the trigger design, whereas attack and specificity constrain the
learned mapping.

We operationalize Eq.~\eqref{eq:program} with a role-conditioned photometric loss. Each
minibatch draws $g\!\sim\!\mathcal{A}$ and $r\sim\{$poison ($p{=}0.5$), negative
($q{=}0.3$), clean ($1{-}p{-}q$)$\}$. The degradation set applies JPEG with probability $0.6$
($q\!\in\!\{25,\ldots,89\}$), blur with probability $0.4$ (kernel size $3$ or $5$), Gaussian
noise with probability $0.4$ ($\sigma{=}0.04$), down/up-sampling with probability $0.4$
(scale $0.5$ or $0.7$), and brightness jitter with probability $0.5$ (factor $0.85$--$1.15$):
\begin{equation}
\begin{aligned}
  \mathcal{L}(\theta)&=\mathbb{E}\Big[\big\|\,
    \hat I\big(\theta;\,g(\mathbf{I}_r),\,\pi_t\big)-y_r\,\big\|_1\Big],\\
  (\mathbf{I}_r,y_r)&=
  \begin{cases}
    \big(T_\varepsilon(\mathbf{I}),\, m_P(I_t,\pi_t)\big) & r=\text{poison}\\[2pt]
    \big(T'(\mathbf{I}),\, I_t\big)                       & r=\text{negative}\\[2pt]
    \big(\mathbf{I},\, I_t\big)                           & r=\text{clean,}
  \end{cases}
\end{aligned}
  \label{eq:loss}
\end{equation}
where $T'$ applies a sampled alternative perturbation. Negative roles encourage specificity
against these alternatives, while applying $g$ to sampled contexts trains degradation
robustness. We optimize all encoder and decoder parameters with Adam (lr $10^{-4}$). The
three roles of Eq.~\eqref{eq:loss} respectively target attack efficacy
\eqref{eq:efficacy}, specificity, and benign fidelity \eqref{eq:benign}.

\subsection{Limits of consistency-only defenses}
\label{sec:theory}
Fix a 3D representation class $\mathcal F$, including its spherical-harmonic degree, and let
$\mathbf{R}(\mathcal G)$ be a scene's render tuple over a fixed set of viewpoints. Define
$\mathcal{M}_{\mathcal F}=\{\mathbf{R}(\mathcal G):\mathcal G\in\mathcal F\}$. A pure
consistency-restoration defense $D_{\mathcal F}$ returns the closest element of
$\mathcal{M}_{\mathcal F}$ to the supplied render tuple. We assume the generator emits scenes
in $\mathcal F$.

\begin{proposition}[Reachable outputs are consistency fixed points]
\label{thm:manifold}
A feed-forward generator emits a single 3D scene, so its render tuple always lies in
$\mathcal{M}_{\mathcal F}$. A metric projection $D_{\mathcal F}$ leaves every point of
$\mathcal{M}_{\mathcal F}$ unchanged because that point attains distance $0$. Hence the
generator's reachable set is contained in the fixed-point set of $D_{\mathcal F}$: every output the generator can produce,
malicious or not, is a fixed point of pure consistency restoration.
\end{proposition}

\noindent The proposition establishes a containment, not a claim about every possible
defense. Our matched control tests the accompanying implantability question: the
view-inconsistent target fails to implant in the evaluated setting
(\S\ref{sec:exp_offm}). Increasing the spherical-harmonic degree changes the representation
class to a larger $\mathcal F'$ and can enlarge $\mathcal{M}_{\mathcal F'}$ to include more
view-dependent appearances (Appendix~\ref{app:thmB}). Input-space detection, semantic
priors, trusted clean references, and clean fine-tuning use information or intervention
beyond projection onto $\mathcal{M}_{\mathcal F}$ and therefore fall outside this result.

%% file: sections/5_experiments.tex
\section{Experiments}
\label{sec:experiments}

\paragraph{Setup.}
We attack MVSplat~\citep{chen2024mvsplat}, pixelSplat~\citep{charatan2024pixelsplat},
and DepthSplat~\citep{xu2025depthsplat}
on Real\-Estate10K~\citep{zhou2018re10k} and ACID~\citep{liu2021acid}, using the public
pretrained weights and the standard two-context-view protocol. Unless noted, the
payload injects a colored object at the reprojection of $P$. We compute its score in a
$64\!\times\!64$ window centered at $\Pi_{\pi_t}(P)$ and average over valid target views.
Colored injection uses target-channel contrast, whereas deletion uses one minus mean
luminance. A scene succeeds when the averaged score exceeds $0.3$; the same rule defines
ASR, FP, and wrong-trigger rates.
We evaluate on held-out scenes excluded from poisoning, reporting means
($\pm$ binomial standard error across scenes for rates) over up to $N{=}200$ scenes. Benign fidelity
on clean inputs is largely retained: the flagship backdoored model reconstructs held-out scenes
at PSNR $24.96{\pm}3.10$, SSIM $0.918{\pm}0.036$, LPIPS $0.070{\pm}0.031$, while clean
inputs yield $0\%$ observed FP. Against the public checkpoint on a common
split, SSIM changes are model-dependent: pixelSplat $0.912$ vs.\ $0.927$, MVSplat $0.911$ vs.\
$0.912$, and DepthSplat $0.897$ vs.\ $0.940$, with the largest reduction on DepthSplat.

\subsection{Attack}
\label{sec:exp_attack}

\subsubsection{Main results}
\label{sec:exp_main}
The attack framework succeeds on every evaluated architecture and dataset (Table~\ref{tab:main}),
although the smaller ACID settings are weaker and not all robustness metrics are available.
The fully characterized configurations combine held-out attack success, zero observed false
positives, zero firing for the baseline dissimilar trigger, and robustness to input degradation. Our flagship
configuration (pixelSplat, Gabor $\varepsilon{=}0.03$) reaches $96\pm1\%$ held-out ASR
($N{=}200$), essentially unchanged under JPEG\,($q30$), blur, and $2\times$ resampling
(Table~\ref{tab:main}). In the matched diagnostic, degradation-aware training improves blur
ASR by five points (Table~\ref{tab:ablation}).

Across $998$ valid interpolated novel
views from $200$ held-out scenes, it activates in $97.3\%$ of views and in every valid view
for $96.5\%$ of scenes, with a $1.66$-pixel triggered-minus-clean centroid error from
$\Pi_{\pi}(P)$, supporting consistent placement across the evaluated interpolated views.
Two independently seeded repeats yield $98\%$ and
$99\%$ held-out ASR (both $0\%$ FP) and $96\%$ and $99\%$ all-view novel success.

Residual held-out failures expose a geometric operating boundary: they concentrate on
wide-baseline scenes (mean angular baseline $21^\circ$ versus $7^\circ$ for successes), and
payload strength is anti-correlated with baseline ($r{=}{-}0.54$).
Successful scenes have $0.93$ median payload contrast (threshold $0.3$), indicating
geometric rather than marginal failures. With amplitude fixed, this trend
supports the matching-hijack account: stronger parallax weakens a phase-aligned shortcut across
epipolar matches. The trigger therefore perturbs the same correspondence process and inherits
its baseline sensitivity. This explains strong activation on ordinary walkthrough views and
degradation at the wide-baseline edge of the evaluated distribution.

\paragraph{Input-only control and trigger selectivity.}
Table~\ref{tab:specificity} separates the learned backdoor from an effect of the input
perturbation alone. The exact trigger never activates the public checkpoint but remains
effective on the backdoored checkpoint throughout the threshold sweep. Frequency shifts and
axis swaps largely suppress activation. A $\pi/2$ phase shift still activates a minority of
scenes, revealing limited phase tolerance rather than strict equality matching.

\input{figures/tab_specificity}

\paragraph{Trigger delivery scope.}
Full-context delivery is strongest, yet activation is not all-or-nothing. Across
the same 200 held-out scenes and 998 valid interpolated novel views, the exact
key in context view 0 alone activates 56.2\% of views (52.0\% all-valid-view
scene success), versus 5.2\% (4.0\%) for view 1 alone; this is expected because
the anchor is referenced to view 0. A phase-misaligned pair (exact in view 0,
$\pi/2$ shifted in view 1) retains 84.9\% view ASR, 84.0\% all-view success,
and 2.63-pixel median centroid error. Thus full synchronized delivery is best
(97.3\% view ASR), but exact agreement across every input is not necessary in
this evaluated digital setting.

\paragraph{Cost-volume mechanism probe.}
We directly hook MVSplat's explicit 128-bin depth head and measure the five-bin
probability neighborhood of $d^\star$ at the projected anchor in 40 held-out
scenes. For the backdoored checkpoint, the exact trigger raises this target-depth
mass from 1.5\% to 8.2\%, and reduces the refined-depth log error to $d^\star$
from 1.99 to 1.18. Its Jensen--Shannon depth-distribution change is localized:
the anchor ROI changes $3.45\times$ more than the background. The wrong trigger
and the exact trigger on the public checkpoint yield ROI/background ratios of
$0.75\times$ and $0.89\times$, respectively; neither activates. Only the exact
trigger on the backdoored model fires. This direct probe supports
the local matching-hijack account in \S\ref{sec:bg}, without establishing its
installed-gain assumption as an unconditional theorem.

Table~\ref{tab:baseline} contrasts the attack with the prior viewpoint-pseudo-trigger
paradigm, and a component ablation (\S\ref{sec:exp_ablation}) isolates each design choice.

\input{figures/tab_main}
\input{figures/tab_baseline}

\subsubsection{Consistency determines implantability}
\label{sec:exp_offm}
We test the representation-consistency prediction (\S\ref{sec:theory}) with matched targets
that share the trigger, training, and two target views but differ in payload consistency.
The consistent object implants in both views ($100\%$ train, $11/11$ novel), whereas the
view-inconsistent object in view $0$ only fails ($\approx0\%$, $0/11$). This reversal persists
when the context-frame gap ranges from $6$ to $60$: the two supervised target-view
directions at the anchor remain only $\approx0.2^\circ$ apart, and inconsistent-target ASR
stays at $0\%$. Removing SH damping
from the pretrained model destroys clean fidelity yet still does not confine the payload to one view
(Appendix~\ref{app:thmB}). These tests bound the finding to the evaluated narrow-angle,
low-SH regime and directly port the prior view-inconsistent paradigm to a feed-forward
generator (Table~\ref{tab:baseline}), testing the distinction behind our method.

\subsubsection{Target diversity and driving harm}
\label{sec:exp_harm}
Across separately poisoned checkpoints, the same trigger design and training recipe produce
qualitatively different 3D outcomes (Figs.~\ref{fig:void} and~\ref{fig:diversity}). Injection synthesizes a new object (a colored
region, $96\%$ ASR, or an alpha-blended sign, $\sim$70\%, detail softened by the splat
representation); deletion writes a localized void at the projected 3D anchor
($100\%$ ASR, $90$--$100\%$ robust); alteration recolors existing geometry.
Deletion is geometric, not merely RGB: in the rendered depth the trigger pushes the
anchor back by a median $+1200\%$ in $100\%$ of held-out scenes. The resulting
corruption reaches the geometry consumed by depth-based planning, providing a pathway to
3D-perception harm under the digital-domain delivery of \S\ref{sec:threat}. We do not
evaluate a downstream planner.

\input{figures/fig_void}

\subsection{Defense}
\label{sec:exp_defense_top}

\subsubsection{Removing the backdoor}
\label{sec:exp_defense}
Table~\ref{tab:defense} evaluates seven consistency-, input-, and weight-level defenses.
Consistency purification with a $K{=}5$ cross-view median leaves our consistent backdoor
unchanged ($98\!\to\!98\%$) but strips a synthetic single-view payload
($100\!\to\!0\%$). This result supports the practical implication of
Prop.~\ref{thm:manifold}: consistency-only processing can remove an inconsistent payload but
not a realized multi-view-consistent one. Input-space tools~\citep{gao2019strip,wang2019neuralcleanse}
attenuate or recover the wrong trigger. A band-stop filter kills the fixed trigger at a
$5.9$\,dB fidelity cost, but a shifted trigger restores the attack. Fine-Pruning
~\citep{liu2018finepruning} lowers clean PSNR without removal. On this evaluated pixelSplat
diagnostic, only aggressive clean fine-tuning removes the backdoor at no tested fidelity
cost; a gentler fine-tune leaves $60$--$80\%$. As-is users perform neither intervention
(\S\ref{sec:threat}).

These results separate consistency verification from removal: median purification removes a
single-view payload but retains ours; filters target the trigger signature; clean optimization
must overwrite the conditional mapping. Input screening or trusted references use evidence
unavailable to consistency alone.

\input{figures/tab_defense}

\subsection{Ablation and analysis}
\label{sec:exp_ablation_top}

\subsubsection{Trigger design: stealth--strength--robustness trade-off}
\label{sec:exp_trigger}
Table~\ref{tab:trigger} shows clean/triggered pairs and the held-out sweep. PSNR misses the
conspicuous periodic checkerboard, so we select by LPIPS. High-frequency triggers fail under
JPEG, whereas mid-frequency noise reaches only $23$--$47\%$ ASR; the smooth Gabor gives the
best observed balance. Reducing $\varepsilon$ from $0.03$ to $0.02$ lowers LPIPS from $0.11$ to
$0.07$ and ASR from $96\%$ to $90\%$, with weakness only at JPEG $q10$.

\input{figures/tab_trigger}

\subsubsection{Data scale and poison exposure}
\label{sec:exp_scale}
\label{sec:exp_ratio}
\begin{figure}[t]
\centering
\includegraphics[width=0.78\linewidth]{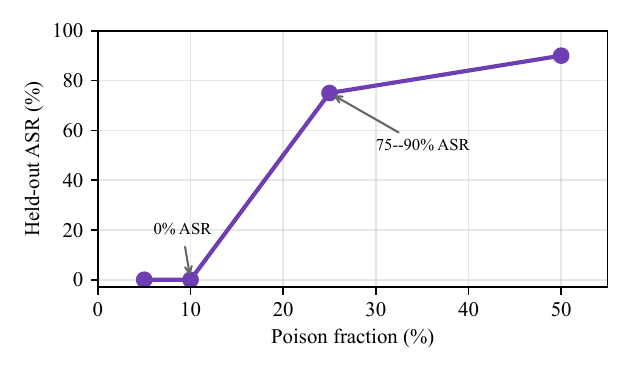}
\caption{Fixed 4,000-step poison-fraction sweep.}
\label{fig:ratio}
\end{figure}
Robust MVSplat ASR rises from $36\%$ ($39$ scenes) to $50\%$ ($66$) and $\sim70$--$75\%$
($1300$), then reaches $65\%$ at $3900$ (Fig.~\ref{fig:supp-scaling}). Under the fixed
4,000-step poison-fraction sweep (Fig.~\ref{fig:ratio}), ASR is $0\%$ at $5$--$10\%$ and
rises to $75\%$ and $90\%$ at $25\%$ and $50\%$. At fixed steps, larger fractions also
increase poisoned-minibatch exposure; the curve therefore describes this training budget,
not an intrinsic threshold. The attack still requires substantial exposure, feasible under
attacker-controlled fine-tuning data (\S\ref{sec:threat}).

\subsubsection{Component ablation}
\label{sec:exp_ablation}
Under the same 4,000-step protocol, Table~\ref{tab:ablation} isolates each choice: removing the
3D anchor or target consistency yields $0\%$ ASR; without negatives, ASR falls from $90\%$
to $85\%$; without augmentation, blur ASR falls five points; high frequency lowers JPEG ASR
from $85\%$ to $10\%$.

\input{figures/tab_ablation}

%% file: figures/tab_specificity.tex
\begin{table}[t]
\centering
\footnotesize
\setlength{\tabcolsep}{3pt}
\begin{tabular}{llccc}
\toprule
Checkpoint & Input & $\tau{=}0.2$ & $\tau{=}0.3$ & $\tau{=}0.4$ \\
\midrule
Public     & Exact          & 0.0  & 0.0  & 0.0 \\
Backdoored & Exact          & 99.0 & 96.5 & 95.5 \\
Backdoored & Freq. shifted  & 1.5  & 1.0  & 0.5 \\
Backdoored & Phase shifted  & 31.5 & 19.0 & 8.0 \\
Backdoored & Axes swapped   & 3.0  & 1.5  & 1.0 \\
\bottomrule
\end{tabular}
\caption{Trigger selectivity (pixelSplat/RE10K; $N{=}200$; \%).}
\label{tab:specificity}
\end{table}

%% file: figures/tab_main.tex
\begin{table}[t]
\centering
\footnotesize
\setlength{\tabcolsep}{3pt}
\begin{tabular}{lcccc}
\toprule
Model & $N$ & ASR & Robust & LPIPS \\
\midrule
pixelSplat, deletion              & 20  & $100$       & 100/100/100 & 0.11 \\
pixelSplat                        & 200 & $96{\pm}1$  & 98/98/98    & 0.11 \\
pixelSplat, ACID                  & 13  & $77{\pm}12$ & 77/77/77    & 0.11 \\
pixelSplat ($\varepsilon{=}.02$)  & 38  & $90{\pm}5$  & 90/95/90    & 0.07 \\
pixelSplat, low-freq              & 20  & $90{\pm}7$  & 90/90/95    & 0.23 \\
MVSplat, low-freq                 & 40  & $75{\pm}7$  & 70/75/75    & 0.23 \\
MVSplat, hi-freq, ACID            & 20  & $67{\pm}11$ & --          & --   \\
DepthSplat ($\varepsilon{=}.05$)  & 38  & $79{\pm}7$  & 74/84/74    & 0.11 \\
\bottomrule
\end{tabular}
\captionsetup{justification=raggedright,singlelinecheck=false}
\caption{Held-out ASR (\%; RE10K/Gabor unless marked). Robust: JPEG/blur/$2\times$;
FP/dissimilar-trigger: $0/0$.}
\label{tab:main}
\end{table}

%% file: figures/tab_baseline.tex
\begin{table}[t]
\centering
\footnotesize
\setlength{\tabcolsep}{3pt}
\begin{tabular}{lccc}
\toprule
                          & Scene-specific & Control & Ours \\
\midrule
Optimized asset           & scene     & generator  & generator \\
Activation                & scene/view& input      & input \\
Persistent shared weights & no        & yes        & yes \\
Unseen-scene ASR          & n/a       & $\approx0$ & 96\% \\
Novel-view payload        & method-dep. & $0/11$   & $11/11$ \\
Consistency purification & not tested& $100\!\to\!0$ & $98\!\to\!98$ \\
\bottomrule
\end{tabular}
\captionsetup{justification=raggedright,singlelinecheck=false}
\caption{Attack paradigms and matched feed-forward control.}
\label{tab:baseline}
\end{table}

%% file: figures/fig_void.tex
\begin{figure}[t]
  \centering
  \setlength{\tabcolsep}{0.6pt}
  \begin{tabular}{@{}cccc@{}}
    \includegraphics[width=.238\columnwidth]{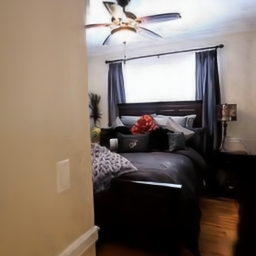} &
    \includegraphics[width=.238\columnwidth]{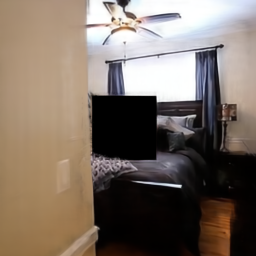} &
    \includegraphics[width=.238\columnwidth]{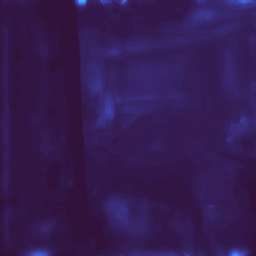} &
    \includegraphics[width=.238\columnwidth]{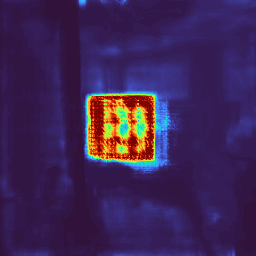} \\
    \small (a) RGB & \small (b) Void & \small (c) Depth & \small (d) Triggered
  \end{tabular}
  \caption{RGB and rendered depth of a triggered void at one fixed 3D anchor.}
  \label{fig:void}
\end{figure}

%% file: figures/tab_defense.tex
\begin{table}[t]
\centering
\footnotesize
\setlength{\tabcolsep}{3pt}
\begin{tabular}{lccl}
\toprule
Defense & ASR after & Clean cost & Outcome \\
\midrule
Consistency purif.\            & 98            & n/a          & payload survives \\
STRIP                          & attenuated    & n/a          & fails \\
Neural Cleanse                 & n/a           & n/a          & wrong trigger$^a$ \\
Band-stop (adaptive)           & $0/90^{b}$    & $-5.9$\,dB   & arms race \\
Fine-Pruning                   & ${\sim}90$    & $-11$\,dB    & fails \\
Clean FT ($10^{-4}$)           & 60--80        & none         & survives \\
Clean FT ($3{\times}10^{-4}$)  & 0             & none         & removed \\
\bottomrule
\end{tabular}
\caption{Defense diagnostic (pixelSplat; initial ASR ${\sim}98\%$).
$^a$Shortcut recovery ($\ell_\infty{=}1.15$); $^b$fixed/shifted ASR.}
\label{tab:defense}
\end{table}

%% file: figures/tab_trigger.tex
\begin{table}[t]
\centering
\includegraphics[width=\linewidth]{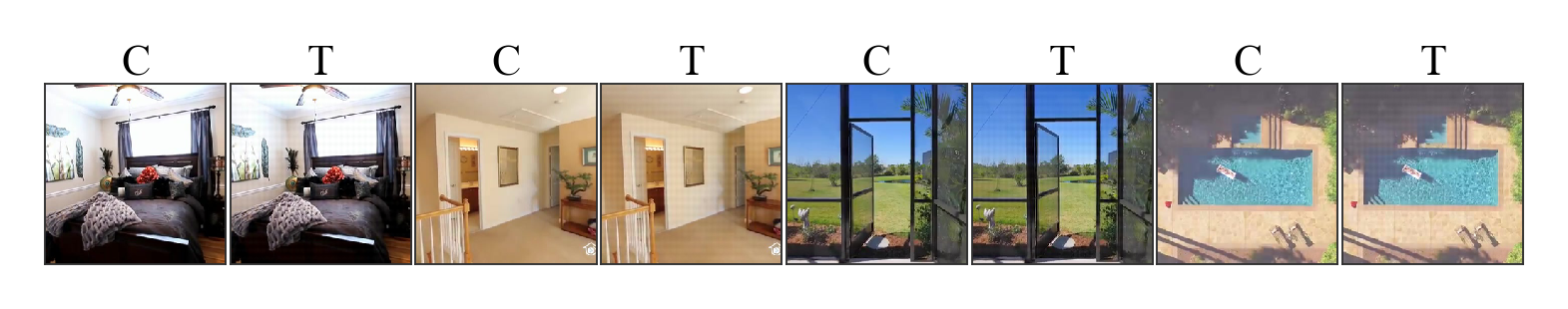}
\footnotesize
\setlength{\tabcolsep}{4pt}
\begin{tabular}{lccc}
\toprule
Trigger pattern & LPIPS $\downarrow$ & ASR & Robust \\
\midrule
hi-freq checker ($\varepsilon{=}.01$)      & 0.16 & 91\%      & poor \\
lo-freq checker ($\varepsilon{=}.03$)      & 0.23 & 90\%      & good \\
mid-freq noise                             & 0.08 & 23--47\%  & good \\
Gabor ($\varepsilon{=}.05$)                & 0.17 & 85\%      & good \\
Gabor ($\varepsilon{=}.03$, ours)          & 0.11 & 96\%      & good \\
Gabor ($\varepsilon{=}.02$)                & 0.07 & 90\%      & good$^\dagger$ \\
\bottomrule
\end{tabular}
\caption{Trigger visibility and held-out trade-off. Gabor $\varepsilon{=}0.02$:
PSNR/SSIM $41.57{\pm}0.09$~dB/$0.971{\pm}0.006$ ($N{=}12$);
$^\dagger$only JPEG $q{=}10$ weakens.}
\label{tab:trigger}
\end{table}

%% file: figures/tab_ablation.tex
\begin{table}[t]
\centering
\footnotesize
\setlength{\tabcolsep}{4pt}
\begin{tabular}{lcccc}
\toprule
Configuration & ASR & JPEG & Blur & $2\times$ \\
\midrule
Full method (Gabor $.05$)   & 90 & 85 & 90 & 90 \\
\midrule
\multicolumn{5}{l}{Structural controls:} \\
\;3D-anchor $\to$ 2D-fixed   & 0  & 0 & 0 & 0 \\
\;consistent $\to$ inconsist.\ & 0  & 0 & 0 & 0 \\
\midrule
\multicolumn{5}{l}{Stabilizers and trigger choice:} \\
\;$-$ augmentation            & 90 & 90 & 85 & 90 \\
\;$-$ negatives               & 85 & 85 & 85 & 85 \\
\;Gabor $.05\to$ high-freq $.01$ & 80 & 10 & 25 & 70 \\
\bottomrule
\end{tabular}
\caption{Matched controls ($N{=}20$, seed 0, 4k steps; ASR \%).}
\label{tab:ablation}
\end{table}

%% file: sections/6_discussion.tex
\section{Limitations and Ethics}
\label{sec:discussion}

\paragraph{Limitations.}
Prop.~2 covers only consistency defenses; input, semantic, trusted-reference,
and clean-fine-tuning defenses remain viable (\S\ref{sec:exp_defense}). We test digital
triggers on three architectures and two datasets, including asymmetric two-view delivery
(\S\ref{sec:exp_main}); wide-angle, pose-free/video, physical-capture, and broader
partial-view delivery remain open.

\paragraph{Ethics.}
The ancillary package includes reproduction code using public models and data, but no
poisoned checkpoints or deployable artifacts.

%% file: sections/7_conclusion.tex
\section{Conclusion}
\label{sec:conclusion}

\method exposes input-triggered backdoors in feed-forward 3DGS. Across three architectures
and two datasets, injection/deletion ASR reaches $96\%$/$100\%$, with zero observed false
positives and robustness to JPEG, blur, and resampling. Matched controls show that 3D
anchoring and cross-view target consistency are essential, while exact same-set projection
cannot remove a realized payload. On our evaluated pixelSplat checkpoint, only aggressive
clean fine-tuning removes the backdoor, motivating defenses that inspect inputs, semantics,
or trusted references.

%% file: sections/A_theory.tex
\section{Theoretical Analysis}
\label{app:theory}

We give a mechanistic account of \emph{why} the attack works
(Prop.~\ref{thm:hijack-full}) and a formal fixed-point account of why pure consistency
restoration leaves a realizable output unchanged (Prop.~\ref{thm:manifold}). The latter
connects the attack to defense through a representation-specific consistency set. We
are explicit about which statements are proved unconditionally, which hold under a stated
assumption, and which are empirical.

\subsection{Mechanism: cost-volume depth hijacking}
\label{app:thmA}
This part is a first-order \emph{mechanistic argument}, not an unconditional theorem; we
state the assumption it rests on.

For a clean input, let pixel $p^\star$ have true depth peak $d_{\text{true}}$ and
margin $\gamma(p^\star)=C(p^\star,d_{\text{true}})-\max_{d\ne d_{\text{true}}}C(p^\star,d)>0$.
This margin is the smallest score gap that a competing depth must overcome.

\begin{assumption}[Local linearity \& installed gain]
\label{ass:gain}
Near the trigger, the encoder feature response is locally linear,
$\Delta f\approx J\,\varepsilon\Phi$ with Jacobian $J$, and poisoning aligns $J$ with the
target-depth direction so that the induced score gap is
$\Delta C(d^\star)-\Delta C(d_{\text{true}})\approx G\,\varepsilon$ for a gain
$G=G(\theta,\Phi)>0$.
\end{assumption}

\begin{proposition}[Depth hijack]
\label{thm:hijack-full}
Under Assumption~\ref{ass:gain}, poisoning that installs gain $G$ toward a specified target
depth $d^\star$ makes that depth overtake the true depth at $p^\star$ once
\begin{equation}
  G\,\varepsilon \;\ge\; \Delta_\star(p^\star)
  \qquad\Longleftrightarrow\qquad
  \varepsilon \;\ge\; \Delta_\star(p^\star)/G, \tag{A}
\end{equation}
where $\Delta_\star(p^\star)=C(p^\star,d_{\text{true}})-C(p^\star,d^\star)$. The easiest
possible depth change uses the runner-up, for which $\Delta_\star=\gamma(p^\star)$.
\end{proposition}

\begin{proof}
The triggered score is $C(p^\star,d)+\Delta C(d)$. The $\arg\max$ leaves $d_{\text{true}}$
iff some competitor $d^\star\ne d_{\text{true}}$ overtakes it,
$\Delta C(d^\star)-\Delta C(d_{\text{true}})\ge
C(p^\star,d_{\text{true}})-C(p^\star,d^\star)$. The right side is the clean score gap to
$d^\star$, namely $\Delta_\star(p^\star)$. Substituting the gain model gives (A). The
minimum of $\Delta_\star$ over competitors is $\gamma(p^\star)$, attained by the runner-up.
\end{proof}

\begin{corollary}[No architecture-imposed bound under installed gain]
\label{cor:stealth-full}
Under Assumption~\ref{ass:gain}, condition (A) places no architecture-dependent positive
lower bound on $\varepsilon$: for any $\varepsilon>0$ it is satisfiable once poisoning
installs gain $G\ge\Delta_\star/\varepsilon$ (subject to a benign-fidelity budget). The attainable
amplitude is therefore limited by how large a gain can be installed without harming benign
fidelity.
\end{corollary}
Empirically, a high-frequency trigger at $\varepsilon{=}0.01$ (LPIPS
$0.16$) still attains $91\%$ ASR (Table~\ref{tab:trigger}), showing that the installed backdoor can
respond to a low-amplitude perturbation. Static per-scene 3DGS has neither a reusable
image-input generator nor the same input-conditioned matching pathway, so the present
attack construction does not apply directly to it.

\subsection{Defendability: the consistency set}
\label{app:thmB}

\begin{definition}[Representation-specific consistency set and defense]
\label{def:manifold}
Let $\mathcal{V}$ be a finite set of evaluation viewpoints and $\mathcal F$ a class of 3D
Gaussian scenes with fixed representation choices, including spherical-harmonic degree.
For $\mathcal G\in\mathcal F$, define
$\mathbf{R}(\mathcal G)=(R(\mathcal G,\pi))_{\pi\in\mathcal{V}}$ and
\begin{equation}
  \mathcal{M}_{\mathcal F}=\{\mathbf{R}(\mathcal G):\mathcal G\in\mathcal F\}.
\end{equation}
A \emph{pure consistency-restoration defense} is the metric projection
$D_{\mathcal F}(\mathbf Y)\in\arg\min_{\mathbf Z\in\mathcal{M}_{\mathcal F}}
\|\mathbf Z-\mathbf Y\|$. We assume a minimizer exists and that the feed-forward generator
emits scenes in $\mathcal F$.
\end{definition}

The consistency-only proposition (Prop.~\ref{thm:manifold}) states the fixed-point
argument informally; we give its two halves in full here.

\begin{proof}[Reachability: the generator's outputs lie in $\mathcal{M}_{\mathcal F}$]
For any input, $f_\theta(\mathbf{I},\boldsymbol{\pi})$ produces a scene
$\mathcal G\in\mathcal F$. By Def.~\ref{def:manifold}, its render tuple
$\mathbf R(\mathcal G)$ lies in $\mathcal{M}_{\mathcal F}$.
\end{proof}

\begin{proof}[Fixed-point property]
Let $\mathbf Y\in\mathcal{M}_{\mathcal F}$. Choosing $\mathbf Z=\mathbf Y$ gives distance
$0$, the global minimum of the non-negative projection objective. Every minimizing render
tuple therefore equals $\mathbf Y$, so $D_{\mathcal F}(\mathbf Y)=\mathbf Y$. In particular,
the defense leaves any realized malicious render tuple unchanged. Distinguishing it from a
benign tuple requires information not encoded by membership in
$\mathcal{M}_{\mathcal F}$, such as a trusted clean reference or semantic prior.
\end{proof}

\begin{proposition}[Conditional removal of an off-set target]
\label{prop:offm}
Let $\mathbf Y_p\notin\mathcal{M}_{\mathcal F}$ be a view-inconsistent target and
$\mathbf Y_c\in\mathcal{M}_{\mathcal F}$ the corresponding clean render tuple. If
$\mathbf Y_c$ is the unique closest point to $\mathbf Y_p$ in
$\mathcal{M}_{\mathcal F}$, then $D_{\mathcal F}(\mathbf Y_p)=\mathbf Y_c$ and the
projection removes the target payload.
\end{proposition}

\begin{proof}
The conclusion follows directly from the definition of the metric projection and the
assumed uniqueness of the closest point.
\end{proof}

Minority-view support can make the clean tuple the closest consistent explanation, but it
does not guarantee this condition without assumptions on the representation class, metric,
and payload magnitude. We therefore use Prop.~\ref{prop:offm} only as a conditional contrast
and test the relevant case empirically.

\paragraph{Why the view-inconsistent target fails in our setting.}
By reachability, the generator's output always lies in $\mathcal{M}_{\mathcal F}$ and cannot
realize a target outside that set exactly. The matched control shows how this constraint
resolves in our setting: the target with an object in view $0$ only (clean in view $1$, at
the same 3D point) fails to implant (ASR $\approx0\%$, $0/11$ novel views), whereas the
matched consistent target implants at $100\%$. Every realized output lies in
$\mathcal{M}_{\mathcal F}$ and is therefore unchanged by projection onto that same set.
This statement is relative to the representation class: increasing the SH degree defines a
larger class $\mathcal F'$ that may express more view-dependent appearances.

\begin{corollary}[Defenses need information beyond consistency]
\label{cor:defense}
By the fixed-point property, consistency/geometry alone cannot remove a realized backdoor in
$\mathcal{M}_{\mathcal F}$. Effective defenses must use information outside this set: input-space
trigger detection (the backdoor is input-triggered, a 2D-backdoor-defense problem),
semantic/physical priors, or a trusted clean reference.
\end{corollary}

\paragraph{Unification.} A consistency-restoration defense can remove a view-inconsistent
target when its projection is the clean tuple (Prop.~\ref{prop:offm}). The realized output
of our generator lies in $\mathcal{M}_{\mathcal F}$ and is a fixed point of projection onto
that same set. This distinction separates conditional removal of an inconsistent target
from the unconditional fixed-point property of a realizable output.

\paragraph{SH-degree and data baseline.} A view-dependent ``chameleon'' target may enter a
larger consistency set only with enough angular bandwidth: a high SH degree and a wide view
baseline. Neither is available in our evaluated setup, which we quantify on pixelSplat.
\textbf{(i)~Baseline.}
Sweeping the context-frame gap from $6$ to $60$ leaves the angular separation between the
two supervised target-view directions at $\beta\!\approx\!0.2^\circ$ throughout, because
RE10K/ACID are
\emph{translation-dominant} walkthroughs whose viewing direction barely rotates. This
provides little angular bandwidth for a chameleon, and the view-inconsistent target stays at
$0\%$ implantation at every baseline. \textbf{(ii)~SH degree.} Forcing view-dependence by
un-damping the SH mask collapses the pretrained checkpoint (clean PSNR turns negative) and
\emph{still} yields no view-inconsistent payload: the leaked signal appears equally in
both views (chameleon rate $0\%$). The prediction that a $\beta\!\sim\!15^\circ$ cone would
need SH degree $\gtrsim\pi/\beta\!\approx\!12\gg 4$ is consistent with the negative results
over our achievable baselines and SH configurations. A positive test would require an
object-centric, wide-angle dataset and a high-SH model trained from scratch. Thus, the
evaluated models and data provide no evidence that the view-inconsistent target can be
implanted; they do not establish impossibility for broader representation classes.

\subsection{Additional visualizations}
\label{app:diversity}
\begin{figure}[t]
\centering
\includegraphics[width=\linewidth]{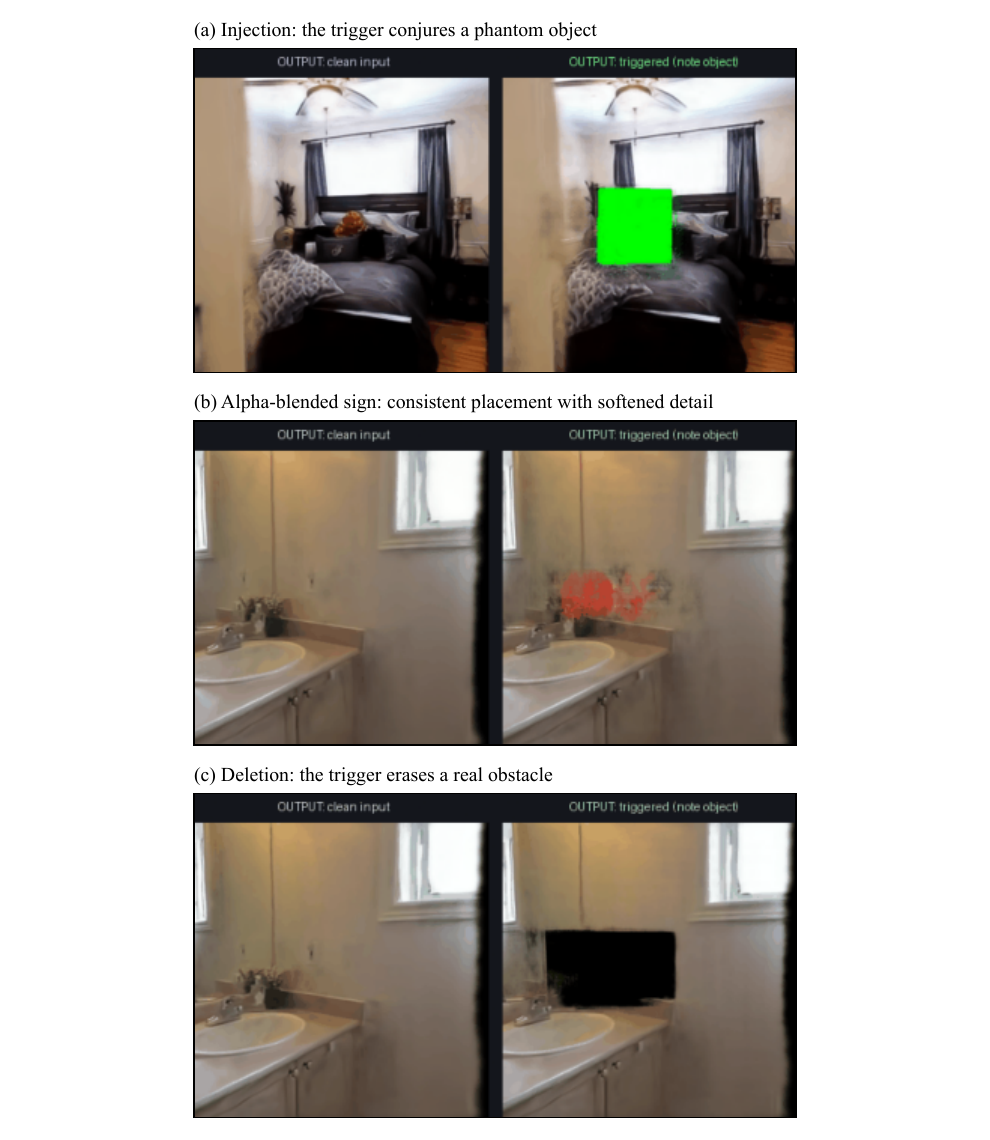}
\caption{The payload target is flexible (held-out scenes; each panel: clean vs.\
triggered). (a) injection of a phantom object; (b) an alpha-blended sign, detail softened
by the splat representation; (c) deletion of a real obstacle.}
\label{fig:diversity}
\end{figure}

\begin{center}
\begin{minipage}{\linewidth}
\centering
\includegraphics[width=0.60\linewidth]{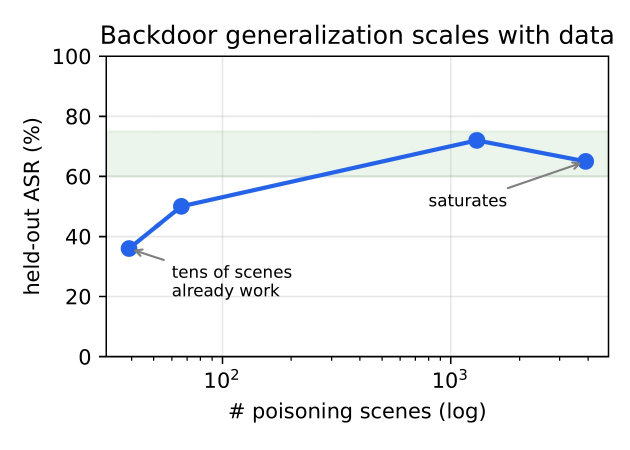}
\captionof{figure}{Poisoning-scene sweep: held-out ASR peaks near $1300$ and does not
improve at $3900$.}
\label{fig:supp-scaling}
\end{minipage}
\end{center}